\documentclass[11pt]{article}

\usepackage{amssymb}
\usepackage{amsfonts}
\usepackage{amsmath}
\usepackage{amsthm}
\usepackage[english]{babel}
\usepackage{latexsym}
\usepackage{color}
\usepackage{enumerate}
\usepackage{nicefrac}
\usepackage{mathrsfs}
\usepackage{comment}
\usepackage[hmargin=2cm,vmargin=2cm]{geometry}
\usepackage{bm}
\usepackage{bbm}
\usepackage{centernot}

\theoremstyle{plain}
\newtheorem{theorem}{Theorem}[section]
\newtheorem{lemma}[theorem]{Lemma}
\newtheorem{proposition}[theorem]{Proposition}

\newtheorem{example}[theorem]{Example}

\theoremstyle{definition}
\newtheorem{definition}[theorem]{Definition}
\newtheorem{assumption}[theorem]{Assumption}
\newtheorem{remark}[theorem]{Remark}

\theoremstyle{remark}

\numberwithin{equation}{section}

\newcommand{\rec}{\mathop{\rm rec}\nolimits}
\newcommand{\qint}{\mathop{\rm qint}\nolimits}

\newcommand{\ES}{\mathop{\rm ES}\nolimits}
\newcommand{\ext}{\mathop{\rm ext}\nolimits}
\newcommand{\lsc}{\mathop{\rm lsc}\nolimits}
\newcommand{\usc}{\mathop{\rm usc}\nolimits}
\newcommand{\barr}{\mathop{\rm bar}\nolimits}

\newcommand{\Closure}{\mathop{\rm cl}\nolimits}

\newcommand{\Interior}{\mathop{\rm int}\nolimits}

\newcommand{\E}{{\mathbb E}}

\newcommand{\N}{{\mathbb N}}

\newcommand{\R}{{\mathbb R}}

\newcommand{\probp}{\mathbb P}

\newcommand{\cA}{{\mathcal{A}}}

\newcommand{\cC}{{\mathcal{C}}}

\newcommand{\cF}{{\mathcal{F}}}

\newcommand{\cU}{{\mathcal{U}}}

\newcommand{\cZ}{{\mathcal{Z}}}

\usepackage[affil-it]{authblk}
\def\keywords{\vspace{.5em}
{\noindent\textbf{Keywords}:\,\relax%
}}

\def\JELclassification{\vspace{.5em}
{\noindent\textbf{JEL classification}:\,\relax%
}}

\def\MSCclassification{\vspace{.5em}
{\noindent\textbf{MSC}:\,\relax%
}}

\makeatletter
\def\@fnsymbol#1{\ensuremath{\ifcase#1\or 1\or 2\or 3\or 4\or 5\or 6\or 7\or 8\else\@ctrerr\fi}}
\makeatother

\begin{document}

\title{Multi-utility representations of incomplete preferences
induced by set-valued risk measures}

\author{Cosimo Munari}
\affil{Center for Finance and Insurance and Swiss Finance Institute \\
Department of Banking and Finance, University of Zurich, Switzerland \\
\texttt{\normalsize cosimo.munari@bf.uzh.ch}}

\date{\today}

\maketitle

\begin{abstract}
We establish a variety of numerical representations of preference relations induced by set-valued risk measures. Because of the general incompleteness of such preferences, we have to deal with multi-utility representations. We look for representations that are both parsimonious (the family of representing functionals is indexed by a tractable set of parameters) and well behaved (the representing functionals satisfy nice regularity properties with respect to the structure of the underlying space of alternatives). The key to our results is a general dual representation of set-valued risk measures that unifies the existing dual representations in the literature and highlights their link with duality results for scalar risk measures.
\end{abstract}

\keywords{risk measures, dual representations, incomplete preferences, multi-utility representations}

\smallskip
\JELclassification{C60, G11}

\smallskip

\MSCclassification{91B06, 91G80}

\parindent 0em \noindent


\section{Introduction}
\label{intro}

This note is concerned with the numerical representation of preference relations induced by a special class of set-valued maps. Recall that a {\em preference (relation)} over the elements of a set $L$ is a reflexive and transitive binary relation on $L$. A preference is said to be {\em complete} if any two elements $x,y\in L$ are comparable in the sense that it is always possible to determine whether $x$ is preferred to $y$ or viceversa. Following the terminology of Dubra et al.~\cite{DubraMaccheroniOk2004}, a family $\cU$ of maps $u:L\to[-\infty,\infty]$ is a {\em multi-utility representation} of a preference $\succeq$ if for all $x,y\in L$ we have
\[
x\succeq y \ \iff \ u(x)\geq u(y) \ \ \mbox{for every} \ u\in\cU.
\]
In words, a multi-utility representation provides a numerical representation for the given preference relation via a family of ``utility functionals''. In view of their greater tractability, multi-utility representations play a fundamental role in applications. A standard problem in this context is to find representations that are at the same time parsimonious (the family of representing functionals is indexed by a small set of parameters) and well behaved (the representing functionals satisfy nice regularity properties with respect to the structure of the underlying set). This is especially important for incomplete preferences, which cannot be represented by a unique functional.

\medskip

In the broad field of economics and finance, incomplete preferences arise naturally in the presence of multi-criteria decision making. We refer to Aumann~\cite{Aumann1962} and Bewley~\cite{Bewley2002} for two classical references and to Ok~\cite{Ok2002}, Dubra et al.~\cite{DubraMaccheroniOk2004}, Mandler~\cite{Mandler2005}, Eliaz and Ok~\cite{EliazOk2006}, Kaminski~\cite{Kaminski2007}, Evren~\cite{Evren2008,Evren2014}, Evren and Ok~\cite{EvrenOk2011}, Bosi and Herden~\cite{Bosiherden2012}, Galaabaatar and Karni~\cite{GalaabaatarKarni2013}, Nishimura and Ok~\cite{NishimuraOk2016}, Bosi et al.~\cite{BosiEstevanZuanon2018}, and Bevilacqua et al.~\cite{BevilacquaBosiKaucicZuanon2019} for an overview of contributions to the theory of incomplete preferences and their multi-utility representations in the last twenty years.

\medskip

The goal of this note is to establish numerical representations of preference relations induced by a special class of set-valued maps that have been the subject of intense research in the recent mathematical finance literature. To introduce the underlying economic problem, consider an economic agent who is confronted with the problem of ranking a number of different alternatives represented by the elements of a set $L$. The agent has specified a target set of acceptable or attractive alternatives $A\subseteq L$. We assume that, if an alternative is not acceptable, it can be made acceptable upon implementation of a suitable admissible action. We represent the results of admissible actions by the elements of a set $M\subseteq L$ and assume that a given alternative $x\in L$ can be transformed through a given $m\in M$ into the new alternative $x+m$. The objective of the agent is then to identify, for each alternative, all the admissible actions that can be implemented to move said alternative inside the target set by way of translations. This naturally leads to the set-valued map $R:L\rightrightarrows M$ defined by
\[
R(x) := \{m\in M \,: \ x+m\in A\} = M\cap(A-x).
\]
The map $R$ can be seen as a generalization of the set-valued risk measures studied by Jouini et al.~\cite{JouiniMeddebTouzi2004}, Kulikov~\cite{Kulikov2008}, Hamel and Heyde~\cite{HamelHeyde2010}, Hamel et al.~\cite{HamelHeydeRudloff2011}, and Molchanov and Cascos~\cite{MolchanovCascos2016} in the context of markets with transaction costs; by Haier et al.~\cite{HaierMolchanovSchmutz2016} in the context of intragroup transfers; by Feinstein et al.~\cite{FeinsteinRudloffWeber2017}, Armenti et al.~\cite{ArmentiCrepeyDrapeauPapapantoleon2018}, and Ararat and Rudloff~\cite{AraratRudloff} in the context of systemic risk. We refer to these contributions for a discussion about the financial interpretation of set-valued risk measures in the respective fields of application and to Section~\ref{sect: applications} for some concrete examples in the context of multi-currency markets with transaction costs and systemic risk.

\medskip

The set-valued map $R$ defined above induces a preference relation on $L$ by setting
\[
x\succeq_R y \ :\iff \ R(x)\supseteq R(y).
\]
According to this preference, the agent prefers $x$ to $y$ if every admissible action through which we can move $y$ into the target set will also allow us to transport $x$ there. In other terms, $x$ is preferred to $y$ if it is easier to make $x$ acceptable compared to $y$. The goal of this note is to establish numerical representations of the preference $\succeq_R$. Since this preference, as shown below, is not complete in general, we have to deal with multi-utility representations. In particular, we look for representations consisting of (semi)continuous utility functionals. We achieve this by establishing suitable (dual) representations of the set-valued map $R$.

\medskip

Our results provide a unifying perspective on the existing dual representations of set-valued risk measures and on the corresponding multi-utility representations, which, to be best of our knowledge, have never been explicitly investigated in the literature. We illustrate the advantages of such a unifying approach by discussing applications to multi-currency markets with transaction costs and systemic risk. In addition, we highlight where our strategy to establishing dual representations differs from the standard arguments used in the literature. The note is structured as follows. The necessary mathematical background is collected in Section~\ref{sect: math background}. The standing assumptions on the space of alternatives and the main properties of the set-valued map under investigation are presented in Section~\ref{sect: setting}. The main results on dual and multi-utility representations are established in Section~\ref{sect: multi utility representations} and are applied to a number of concrete situations in Section~\ref{sect: applications}.


\section{Mathematical background}
\label{sect: math background}

In this section we collect the necessary mathematical background and fix the notation and terminology used throughout the paper. We refer to Rockafellar~\cite{Rockafellar1974} and Z\u{a}linescu~\cite{Zalinescu2002} for a thorough presentation of duality for topological vector spaces. Moreover, we refer to Aubin and Ekeland~\cite{AubinEkeland1984} for a variety of results on support functions and barrier cones.

\medskip

Let $L$ be a real locally convex Hausdorff topological vector space. The topological dual of $L$ is denoted by $L'$. Any linear subspace $M\subseteq L$ is canonically equipped with the relative topology inherited from $L$. The corresponding dual space is denoted by $M'$. For every set $A\subseteq L$ we denote by $\Interior(A)$ and $\Closure(A)$ the interior and the closure of $A$, respectively. We say that $A$ is convex if $\lambda A+(1-\lambda) A\subseteq A$ for every $\lambda\in[0,1]$ and that $A$ is a cone if $\lambda A\subseteq A$ for every $\lambda\in[0,\infty)$. The (lower) support function of $A$ is the map $\sigma_A:L'\to[-\infty,\infty]$ defined by
\[
\sigma_A(\psi) := \inf_{x\in A}\psi(x).
\]
The effective domain of $\sigma_A$ is called the barrier cone of $A$ and is denoted by
\[
\barr(A) := \{\,\psi\in L' \,: \ \sigma_A(\psi)>-\infty\}.
\]
It follows from the Hahn-Banach Theorem that, if $A$ is closed and convex, then it can be represented as the intersection of all the halfspaces containing it or equivalently
\begin{equation}
\label{eq: external representation}
A = \bigcap_{\psi\in L'}\{x\in L \,: \ \psi(x)\geq\sigma_A(\psi)\} = \bigcap_{\psi\in\barr(A)}\{x\in L \,: \ \psi(x)\geq\sigma_A(\psi)\}.
\end{equation}
If $A$ is a cone, then $\barr(A)$ coincides with the polar or dual cone of $A$, i.e.
\[
\barr(A) = A^+ := \{\psi\in L' \,: \ \psi(x)\geq0, \ \forall x\in A\}.
\]
If $A$ is a vector space, then $\barr(A)$ coincides with the annihilator of $A$, i.e.
\[
\barr(A) = A^\bot := \{\psi\in L' \,: \ \psi(x)=0, \ \forall x\in A\}.
\]
Finally, if $A+K\subseteq A$ for some cone $K\subseteq L$, then $\barr(A)\subseteq K^+$.


\section{The setting}
\label{sect: setting}

Throughout the remainder of the note, we assume that $L$ is a real locally convex Hausdorff topological vector space. We also fix a closed convex cone $K\subseteq L$ satisfying $K-K=L$ and consider the induced partial order defined by
\[
x\succeq_Ky \ :\iff \ x-y\in K.
\]
The above partial order is meant to capture an ``objective'' preference relation shared by all agents. This is akin to the ``better for sure'' preference in Drapeau and Kupper~\cite{DrapeauKupper2013}.

\begin{assumption}
We stipulate the following assumptions on $A$ and $M$:
\begin{enumerate}
  \item[(A1)] $A$ is closed, convex, and satisfies $A+K\subseteq A$.
  \item[(A2)] $M$ is a closed linear subspace of $L$ such that $M\cap K\neq\{0\}$.
  \item[(A3)] $R(x)\notin\{\emptyset,M\}$ for some $x\in L$.
\end{enumerate}
\end{assumption}

\smallskip

The next proposition collects a number of basic properties of the set-valued map $R$ and its associated preference $\succeq_R$. The properties of $R$ are aligned with those discussed in Hamel and Heyde~\cite{HamelHeyde2010} and Hamel et al.~\cite{HamelHeydeRudloff2011}.

\begin{proposition}
\label{general properties}
\begin{enumerate}[(i)]
  \item $\succeq_R$ is monotone with respect to $K$, i.e.\ for all $x,y\in L$
\[
x\succeq_Ky \ \implies \ x\succeq_Ry.
\]
  \item $\succeq_R$ is convex, i.e.\ for all $x,y\in L$ and $\lambda\in[0,1]$
\[
x\succeq_Ry \ \implies \ \lambda x+(1-\lambda)y\succeq_Ry.
\]
  \item $R(x)+K\cap M\subseteq R(x)$ for every $x\in L$.
  \item $R(x+m)=R(x)-m$ for all $x\in L$ and $m\in M$.
  \item $R(\lambda x+(1-\lambda)y)\supseteq\lambda R(x)+(1-\lambda)R(y)$ for all $x,y\in L$ and $\lambda\in[0,1]$.
  \item $R(x)$ is convex and closed for every $x\in L$.
  \item $R(x)\neq M$ for every $x\in L$.
\end{enumerate}
\end{proposition}
\begin{proof}
To establish (i), assume that $x\succeq_Ky$ for $x,y\in L$. For every $m\in R(y)$ we have
\[
x+m=y+m+x-y\in A+K\subseteq A.
\]
This shows that $m\in R(x)$ as well, so that $x\succeq_Ry$. To establish (ii), take $\lambda\in[0,1]$ and assume that $x\succeq_Ry$. For every $m\in R(y)$ we have $y+m\in A$ and, hence, $x+m\in A$. This yields
\[
\lambda x+(1-\lambda)y+m=\lambda(x+m)+(1-\lambda)(y+m)\in \lambda A+(1-\lambda)A \subseteq A,
\]
showing that $m\in R(\lambda x+(1-\lambda)y)$. In sum, $\lambda x+(1-\lambda)y\succeq_R y$. To see that properties (iii) to (vi) hold, it suffices to recall that $R(x)=M\cap(A-x)$ for every $x\in L$. Finally, to establish (vii), assume that $R(x)=M$ for some $x\in L$. Take any $y\in L$ and assume that $R(y)$ is nonempty so that $y+m\in A$ for some $m\in M$. For all $n\in M$ and $\lambda\in(0,1]$ we have
\[
\lambda\bigg(x+\frac{1}{\lambda}(n-m)\bigg)+(1-\lambda)(y+m) \in \lambda A+(1-\lambda)A \subseteq A
\]
by convexity. Hence, letting $\lambda\to0$, we obtain $y+n\in A$ by closedness. Since $n$ was arbitrary, we infer that $R(y)=M$. This contradicts assumption (A3), showing that $R(x)\neq M$ must hold for every $x\in L$.
\end{proof}

\smallskip

\begin{remark}
(i) If $M$ is spanned by a single element, then $\succeq_R$ is complete. Indeed, in this case, we can always assume that $M$ is spanned by a nonzero element $m\in M\cap K$ by our standing assumption. Then, for every $x\in L$ such that $R(x)\neq\emptyset$ we see that
\[
R(x) = \{\lambda m \,: \ \lambda\in[\lambda_x,\infty)\}
\]
for a suitable $\lambda_x\in\R$. This shows that $\succeq_R$ is complete.

\smallskip

(ii) In general, the preference $\succeq_R$ is not complete when $M$ is spanned by more than one element. For instance, let $L=\R^3$ and assume that $K=A=\R^3_+$ and $M=\R^2\times\{0\}$. For $x=0$ and $y=(1,-1,0)$ we respectively have
\[
R(x)=\{m\in M \,: \ m_1\geq0, \ m_2\geq0\} \ \ \ \mbox{and} \ \ \ R(y)=\{m\in M \,: \ m_1\geq-1, \ m_2\geq1\}.
\]
Clearly, neither $x\succeq_Ry$ nor $y\succeq_Rx$ holds, showing that $\succeq_R$ is not complete.

\smallskip

(iii) Sometimes the preference $\succeq_R$ is complete even if $M$ is spanned by more than one element. For instance, let $L=\R^3$ and assume that $K=A=\R^2_+\times\R$ and $M=\{0\}\times\R^2$. For every $x\in L$ such that $R(x)\neq\emptyset$ we have
\[
R(x) = \{m\in M \,: \ m_2\geq-x_2\}.
\]
This shows that $\succeq_R$ is complete.
\end{remark}


\section{Multi-utility representations}
\label{sect: multi utility representations}

In this section we establish a variety of multi-utility representations of the preference induced by $R$, which are derived from suitable representations of the sets $R(x)$. As highlighted below, both representations have a strong link with (scalar) risk measures and their dual representations. We refer to the appendix for the necessary mathematical background and notation.

\medskip

The first multi-utility representation is based on the following scalarizations of $R$. Here, we set
\[
K_M^+ := \{\pi\in M' \,: \ \pi(m)\geq0, \ \forall m\in K\cap M\}.
\]

\smallskip

\begin{definition}
For every $\pi\in K_M^+$ we define a map $\rho_\pi:L\to[-\infty,\infty]$ by setting
\[
\rho_\pi(x) := \inf\{\pi(m) \,: \ m\in M, \ x+m\in A\}.
\]
Moreover, we define a map $u_\pi:L\to[-\infty,\infty]$ by setting
\[
u_\pi(x) := -\rho_\pi(x).
\]
\end{definition}

\smallskip

The functionals $\rho_\pi$ are examples of the risk measures introduced in F\"{o}llmer and Schied~\cite{FoellmerSchied2002} and generalized in Frittelli and Scandolo~\cite{FrittelliScandolo2006}. We refer to Farkas et al.~\cite{FarkasKochMunari2014,FarkasKochMunari2015} for a thorough investigation of such functionals at our level of generality. The next proposition features some of their standard properties, which follow immediately from Proposition~\ref{general properties}. Since the announced multi-utility representation will be expressed in terms of the negatives of the functionals $\rho_\pi$, the proposition is stated in terms of the utility functionals $u_\pi$.

\begin{proposition}
For every $\pi\in K_M^+$ the functional $u_\pi$ satisfies the following properties:
\begin{enumerate}[(i)]
  \item $u_\pi$ is translative along $M$, i.e.\ for all $x\in L$ and $m\in M$
\[
u_\pi(x+m) = u_\pi(x)+\pi(m).
\]
  \item $u_\pi$ is nondecreasing with respect to $\succeq_K$, i.e.\ for all $x,y\in L$
\[
x\succeq_K y \ \implies \ u_\pi(x)\geq u_\pi(y).
\]
  \item $u_\pi$ is concave, i.e.\ for all $x,y\in L$ and $\lambda\in[0,1]$
\[
u_\pi(\lambda x+(1-\lambda)y) \geq \lambda u_\pi(x)+(1-\lambda)u_\pi(y).
\]
\end{enumerate}
\end{proposition}

\smallskip

\begin{remark}
Note that, unless $M$ is spanned by one element, the closedness of the set $A$ is not sufficient to ensure that the functionals $\rho_\pi$ are lower semicontinuous; see Example 1 in Farkas et al.~\cite{FarkasKochMunari2015}. We refer to Hamel et al.~\cite{HamelHeydeLoehneRudloffSchrage2015} for a discussion on general sufficient conditions ensuring the lower semicontinuity of scalarizations of set-valued maps and to Farkas et al.~\cite{FarkasKochMunari2015} and Baes et al.~\cite{BaesKochMunari2020} for a variety of sufficient conditions in a risk measure setting.
\end{remark}

\smallskip

The first multi-utility representation of the preference induced by $R$ rests on the intimate link between the risk measures $\rho_\pi$ and the support functions corresponding to $R$.

\begin{lemma}
\label{theo: dual representation 1}
For every $x\in L$ the set $R(x)$ can be represented as
\[
R(x) = \bigcap_{\pi\in K_M^+\setminus\{0\}}\{m\in M \,: \ \pi(m)\geq\rho_\pi(x)\}.
\]
\end{lemma}
\begin{proof}
The result is clear if $R(x)=\emptyset$. Otherwise, recall that $R(x)$ is closed and convex by Proposition~\ref{general properties} and observe that $\rho_\pi(x)=\sigma_{R(x)}(\pi)$ for every $\pi\in M'$. We can apply the dual representation~\eqref{eq: external representation} in the context of the space $M$ to obtain
\[
R(x) = 
\bigcap_{\pi\in\barr(R(x))\setminus\{0\}}\{m\in M \,: \ \pi(m)\geq\rho_\pi(x)\}.
\]
As $R(x)+K\cap M\subseteq R(x)$ again by Proposition~\ref{general properties}, we conclude by noting that the barrier cone of $R(x)$ must be contained in $K_M^+$.
\end{proof}

\smallskip

\begin{theorem}
\label{cor: first multi utility representation}
The preference $\succeq_R$ can be represented by the multi-utility family
\[
\cU=\{u_\pi \,: \ \pi\in K_M^+\setminus\{0\}\}.
\]
\end{theorem}
\begin{proof}
We rely on Lemma~\ref{theo: dual representation 1}. Take any $x,y\in L$. If $x\succeq_Ry$, then $R(x)\supseteq R(y)$ and
\[
\rho_\pi(x) = \sigma_{R(x)}(\pi) \leq \sigma_{R(y)}(\pi) = \rho_\pi(y)
\]
for every $\pi\in K_M^+\setminus\{0\}$. Conversely, if $\rho_\pi(x)\leq\rho_\pi(y)$ for every $\pi\in K_M^+\setminus\{0\}$, then for each $m\in R(y)$ we have $\pi(m) \geq \rho_\pi(y) \geq \rho_\pi(x)$ for every $\pi\in K_M^+\setminus\{0\}$, so that $m\in R(x)$. This yields $x\succeq_R y$ and concludes the proof.
\end{proof}

\smallskip

\begin{remark}
The simple representation in Lemma~\ref{theo: dual representation 1} shows that the set-valued map $R$ is completely characterized by the family of functionals $\rho_\pi$. In the context of risk measures, one could say that a set-valued risk measure is completely characterized by the corresponding family of scalar risk measures. This corresponds to the ``setification'' formula in Section 4.2 in Hamel et al.~\cite{HamelHeydeLoehneRudloffSchrage2015}.
\end{remark}

\smallskip

We aim to improve the above representation in a twofold way. First, we want to find a multi-utility representation consisting of a smaller number of representing functionals. This is important to ensure a more parsimonious, hence tractable, representation. Second, we want to establish a multi-utility representation consisting of (semi)continuous representing functionals. This is important in applications, e.g.\ in optimization problems where the preference appears in the optimization domain.

\medskip

The second multi-utility representation will be expressed in terms of the following utility functionals. Here, for any functional $\pi\in M'$ we denote by $\ext(\pi)$ the set of all linear continuous extensions of $\pi$ to the whole space $L$, i.e.
\[
\ext(\pi) := \{\psi\in L' \,: \ \psi(m)=\pi(m), \ \forall m\in M\}.
\]

\smallskip

\begin{definition}
For every $\pi\in K_M^+$ we define a map $\rho_\pi^\ast:L\to[-\infty,\infty]$ by setting
\[
\rho_\pi^\ast(x) = \sup_{\psi\in \ext(\pi)}\{\sigma_A(\psi)-\psi(x)\} =
\sup_{\psi\in\ext(\pi)\cap\barr(A)}\{\sigma_A(\psi)-\psi(x)\}.
\]
Moreover, we define a map $u_\pi^\ast:L\to[-\infty,\infty]$ by setting
\[
u_\pi^\ast(x) = -\rho_\pi^\ast(x).
\]
(If $A$ is a cone, then $\sigma_A=0$ on $\barr(A)$ and the above maps simplify accordingly).
\end{definition}

\smallskip

The functionals $\rho_\pi^\ast$ are inspired by the dual representation of the risk measures $\rho_\pi$, see e.g.\ Frittelli and Scandolo~\cite{FrittelliScandolo2006} or Farkas et al.~\cite{FarkasKochMunari2015}. The precise link is shown in Proposition~\ref{prop: lsc hull} below. For the time being, we are interested in highlighting some properties of the functionals $\rho_\pi^\ast$, or equivalently $u_\pi^\ast$, and proceeding to our desired multi-utility representation.

\begin{proposition}
For every $\pi\in K_M^+$ the functional $u_\pi^\ast$ satisfies the following properties:
\begin{enumerate}[(i)]
  \item $u_\pi^\ast$ is translative along $M$, i.e.\ for all $x\in L$ and $m\in M$
\[
u^\ast_\pi(x+m) = u^\ast_\pi(x)+\pi(m).
\]
  \item $u^\ast_\pi$ is nondecreasing with respect to $\succeq_K$, i.e.\ for all $x,y\in L$
\[
x\succeq_K y \ \implies \ u^\ast_\pi(x)\geq u^\ast_\pi(y).
\]
  \item $u^\ast_\pi$ is concave, i.e.\ for all $x,y\in L$ and $\lambda\in[0,1]$
\[
u^\ast_\pi(\lambda x+(1-\lambda)y) \geq \lambda u^\ast_\pi(x)+(1-\lambda)u^\ast_\pi(y).
\]
  \item $u^\ast_\pi$ is upper semicontinuous, i.e.\ for every net $(x_\alpha)\subseteq L$ and every $x\in L$
\[
x_\alpha\to x \ \implies \ \limsup u^\ast_\pi(x_\alpha)\geq u^\ast_\pi(x).
\]
\end{enumerate}
\end{proposition}
\begin{proof}
Translativity follows from the definition of $\rho^\ast_\pi$. Being a supremum of affine maps, it is clear that $\rho^\ast_\pi$ is convex and lower semicontinuous. To show monotonicity, it suffices to observe that $\barr(A)\subseteq K^+$ by (A1) and therefore
\[
\rho_\pi^\ast(x) = \sup_{\psi\in\ext(\pi)\cap K^+}\{\sigma_A(\psi)-\psi(x)\} \leq \sup_{\psi\in\ext(\pi)\cap K^+}\{\sigma_A(\psi)-\psi(y)\} = \rho_\pi^\ast(y)
\]
for all $x,y\in L$ with $x\succeq_Ky$, where we used that $\psi(x-y)\geq0$ for every $\psi\in K^+$.
\end{proof}

\smallskip

To streamline the proof of the announced multi-utility representation, we start with the following lemma. We denote by $\ker(\pi)$ the kernel of $\pi\in M'$, i.e.
\[
\ker(\pi) := \{m\in M \,: \ \pi(m)=0\}.
\]
In the sequel, we will repeatedly use the fact that $\ker(\pi)$ has codimension $1$ in $M$ (provided $\pi$ is nonzero).

\begin{lemma}
\label{lem: properties augmented set}
The set $A$ can be represented as
\begin{equation}
\label{eq: properties augmented set}
A = \bigcap_{\pi\in K_M^+\setminus\{0\}}\Closure(A+\ker(\pi)).
\end{equation}
Moreover, for every $\pi\in K_M^+$ we have $\barr(\Closure(A+\ker(\pi)))=\barr(A)\cap\ker(\pi)^\bot$ and
\[
\sigma_{\Closure(A+\ker(\pi))}(\psi)=
\begin{cases}
\sigma_A(\psi) & \mbox{if $\psi\in\ker(\pi)^\bot$},\\
-\infty & \mbox{otherwise}.
\end{cases}
\]
\end{lemma}
\begin{proof}
We only prove the inclusion ``$\supseteq$'' in~\eqref{eq: properties augmented set} because the other assertions are clear. Assume that $x\in\Closure(A+\ker(\pi))$ for every nonzero $\pi\in K_M^+$ and take any $\psi\in\barr(A)$. In light of~\eqref{eq: external representation}, to conclude the proof we have to show that $\psi(x)\geq\sigma_A(\psi)$. To this effect, let $\pi_\psi$ be the restriction of $\psi$ to the space $M$. Since $\barr(A)\subseteq K^+$, it follows that $\pi_\psi\in K_M^+$. Moreover, note that $\psi\in\ker(\pi_\psi)^\bot$. As a result, we have
\[
\psi \in \barr(A)\cap\ker(\pi_\psi)^\bot = \barr(A+\ker(\pi_\psi)) = \barr(\Closure(A+\ker(\pi_\psi))).
\]
Since $x\in\Closure(A+\ker(\pi_\psi))$ by our assumption, we can use~\eqref{eq: external representation} again to get
\[
\psi(x) \geq \sigma_{\Closure(A+\ker(\pi_\psi))}(\psi) = \sigma_{A+\ker(\pi_\psi)}(\psi) = \sigma_A(\psi),
\]
where the last equality holds because $\psi\in\ker(\pi_\psi)^\bot$. This concludes the proof.
\end{proof}

\smallskip

The next lemma records a representation of the map $R$ that will immediately yield our desired multi-utility representation with (upper) semicontinuous functionals.

\begin{lemma}[{\bf Dual representation of $R$}]
\label{theo: dual representation 2}
For every $x\in L$ the set $R(x)$ can be represented as
\begin{align*}
R(x)
&=
\bigcap_{\pi\in K_M^+\setminus\{0\}}\{m\in M \,: \ \pi(m)\geq\rho_\pi^\ast(x)\} \\
&=
\bigcap_{\pi\in K_M^+\setminus\{0\}}\bigcap_{\psi\in \ext(\pi)}\{m\in M \,: \ \pi(m)\geq\sigma_A(\psi)-\psi(x)\}.
\end{align*}
(If $A$ is a cone, then $\sigma_A=0$ on $\barr(A)$ and the representation simplifies accordingly).
\end{lemma}
\begin{proof}
Fix $x\in L$. It follows from the representation in~\eqref{eq: external representation} and Lemma~\ref{lem: properties augmented set} that
\begin{equation}
\label{eq: dual set valued A closed}
R(x) =
\bigcap_{\pi\in K_M^+\setminus\{0\}}\bigcap_{\psi\in\ker(\pi)^\bot}\{m\in M \,: \ \psi(m)\geq\sigma_A(\psi)-\psi(x)\}.
\end{equation}
To establish the desired representation of $R(x)$ it then suffices to show that the set $\ker(\pi)^\bot$ in the right-hand side of \eqref{eq: dual set valued A closed} can be replaced by $\ext(\pi)$. To this effect, let $m\in M$ satisfy $\pi(m)\geq\sigma_A(\psi)-\psi(x)$ for all nonzero $\pi\in K_M^+$ and $\psi\in\ext(\pi)$. Moreover, take an arbitrary nonzero $\pi\in K_M^+$ and an arbitrary $\psi\in\ker(\pi)^\bot$. To conclude the proof, we have to show that
\begin{equation}
\label{eq: dual set valued A closed 3}
\psi(m) \geq \sigma_A(\psi)-\psi(x).
\end{equation}
This is clear if $\psi\notin\barr(A)$ or $\psi\in\ext(\pi)$. Hence, assume that $\psi\in\barr(A)\setminus\ext(\pi)$. Note that, since $\pi$ is nonzero and $K-K=L$, we find $n\in K_M$ such that $\pi(n)>0$. Since $\barr(A)\subseteq K^+$, two situations are possible. On the one hand, if $\psi(n)>0$, then $\psi$ belongs to $\ext(\pi)$ up to a strictly-positive multiple and therefore~\eqref{eq: dual set valued A closed 3} holds. On the other hand, if $\psi(n)=0$, then we must have $\psi\in M^\bot$. To deal with this case, note first that we always find a nonzero $\pi^\ast\in K_M^+$ satisfying $\ext(\pi^\ast)\cap\barr(A)\neq\emptyset$, for otherwise every functional in $\barr(A)\cap\ker(\pi^\ast)^\bot$ would annihilate the entire $M$ and it would follow from~\eqref{eq: external representation} and~\eqref{eq: dual set valued A closed} that $R(y)=M$ for every $y\in A$, which is against Proposition~\ref{general properties}. Now, take $\varphi\in\ext(\pi^\ast)\cap\barr(A)$ and set $\varphi_k=\varphi+k\psi\in\ext(\pi^\ast)$ for each $k\in\N$. It follows that
\begin{align*}
\pi^\ast(m)
&=
\sup_{k\in\N}\varphi_k(m)
\geq
\sup_{k\in\N}\{\sigma_A(\varphi_k)-\varphi_k(x)\} \\
&\geq
\sigma_A(\varphi)-\varphi(x)+\sup_{k\in\N}\{k(\sigma_A(\psi)-\psi(x))\}.
\end{align*}
This implies that $\psi(m)=0\geq\sigma_A(\psi)-\psi(x)$ must hold, establishing~\eqref{eq: dual set valued A closed 3}.
\end{proof}

\smallskip

\begin{theorem}
\label{cor: second multi utility representation}
The preference $\succeq_R$ can be represented by the multi-utility family
\[
\cU^\ast=\{u^\ast_\pi \,: \ \pi\in K_M^+\setminus\{0\}, \ \ext(\pi)\cap\barr(A)\neq\emptyset\}.
\]
\end{theorem}
\begin{proof}
Note that $\rho^\ast_\pi(x)=-\infty$ for every $x\in L$ whenever $\ext(\pi)\cap\barr(A)=\emptyset$ for some $\pi\in K_M^+$. Hence, the desired assertion follows immediately from Lemma~\ref{theo: dual representation 2}; see also the proof of Theorem~\ref{cor: first multi utility representation}.
\end{proof}

\smallskip

\begin{remark}
The statement of Lemma~\ref{theo: dual representation 2} provides a unifying formulation for the dual representations of set-valued risk measures in the literature. This is further illustrated in Section~\ref{sect: applications}. The strategy used in the proof is different from the ones adopted in the literature, which are often based on results from set-valued duality, thereby offering a complementary perspective on the existing proofs; see also Remark~\ref{rem: strategy}.
\end{remark}

\smallskip

The next proposition shows the link between the two multi-utility representations we have established. In a sense made precise below, the representation $\cU^\ast$ can be seen as the regularization of $\cU$ by means of (upper) semicontinuous hulls. Before we show this, it is useful to single out the following dual representation of the augmented acceptance set, which should be compared with Theorem 1 in Farkas et al.~\cite{FarkasKochMunari2015}.

\begin{lemma}
\label{lem: representation augmented}
For every $\pi\in K_M^+\setminus\{0\}$ such that $\ext(\pi)\cap\barr(A)\neq\emptyset$ we have
\[
\Closure(A+\ker(\pi)) = \bigcap_{\psi\in\ext(\pi)}\{x\in L \,: \ \psi(x)\geq\sigma_A(\psi)\}.
\]
\end{lemma}
\begin{proof}
In view of~\eqref{eq: external representation} and Lemma~\ref{lem: properties augmented set}, the assertion is equivalent to
\[
\bigcap_{\psi\in\ker(\pi)^\bot}\{x\in L \,: \ \psi(x)\geq\sigma_A(\psi)\} = \bigcap_{\psi\in\ext(\pi)}\{x\in L \,: \ \psi(x)\geq\sigma_A(\psi)\}.
\]
We only need to show the inclusion ``$\supseteq$''. To this end, we mimic the argument in the proof of Lemma~\ref{theo: dual representation 2}. Let $x\in L$ belong to the right-hand side above and take $\psi\in\ker(\pi)^\bot$. We have to show that
\begin{equation}
\label{eq: representation augmented}
\psi(x)\geq\sigma_A(\psi).
\end{equation}
This is clear if $\psi\notin\barr(A)$ or $\psi\in\ext(\pi)$. Hence, assume that $\psi\in\barr(A)\setminus\ext(\pi)$. Note that, since $\pi$ is nonzero and $K-K=L$, we find $n\in K_M$ such that $\pi(n)>0$. Since $\barr(A)\subseteq K^+$, two situations are possible. On the one hand, if $\psi(n)>0$, then $\psi$ belongs to $\ext(\pi)$ up to a strictly-positive multiple and therefore~\eqref{eq: representation augmented} holds. On the other hand, if $\psi(n)=0$, then we must have $\psi\in M^\bot$. In this case, take any functional $\varphi\in\ext(\pi)\cap\barr(A)$ and set $\varphi_k=k\psi+\varphi\in\ext(\pi)$ for every $k\in\N$. Then,
\[
\psi(x)+\frac{1}{k}\varphi(x) = \frac{1}{k}\varphi_k(x) \geq \frac{1}{k}\sigma_A(\varphi_k) \geq \sigma_A(\psi)+\frac{1}{k}\sigma_A(\varphi)
\]
for every $k\in\N$. Letting $k\to\infty$ yields~\eqref{eq: representation augmented} and concludes the proof.
\end{proof}

\smallskip

For a given map $f:L\to[-\infty,\infty]$ we denote by $\lsc(f)$ the largest lower semicontinuous map dominated by $f$ and, similarly, by $\usc(f)$ the smallest upper semicontinuous map dominating $f$.

\begin{proposition}
\label{prop: lsc hull}
For every $\pi\in K_M^+\setminus\{0\}$ such that $\ext(\pi)\cap\barr(A)\neq\emptyset$ the following statements hold:
\begin{enumerate}[(i)]
  \item $\rho^\ast_\pi=\lsc(\rho_\pi)$.
  \item $u^\ast_\pi=\usc(u_\pi)$.
\end{enumerate}
\end{proposition}
\begin{proof}
Fix a nonzero $\pi\in K_M^+$ such that $\ext(\pi)\cap\barr(A)\neq\emptyset$. Clearly, we only need to show (i). To this effect, recall that $\rho^\ast_\pi$ is lower semicontinuous and note that it is dominated by $\rho_\pi$. Indeed, for every $x\in L$ and for every $m\in M$ such that $x+m\in A$
\[
\sup_{\psi\in\ext(\pi)}\{\sigma_A(\psi)-\psi(x)\} \leq \sup_{\psi\in\ext(\pi)}\{\psi(x+m)-\psi(x)\} = \pi(m),
\]
showing that $\rho^\ast_\pi(x)\leq\rho_\pi(x)$. Now, take a lower semicontinuous map $f:L\to[-\infty,\infty]$ such that $f\leq\rho_\pi$. We claim that $f\leq\rho_\pi^\ast$ as well. To show this, suppose to the contrary that $f(x)>\rho_\pi^\ast(x)$ for some $x\in L$. Note that
\[
\rho_\pi^\ast(x) = \inf\{\lambda\in\R \,: \ x+\lambda m\in\Closure(A+\ker(\pi))\}
\]
by Lemma~\ref{lem: representation augmented}, where $m\in M$ is any element satisfying $\pi(m)=1$ (which exists because $\pi$ is nonzero). As a result, we must have $f(x)>\lambda$ for some $\lambda\in\R$ such that $x+\lambda m\in\Closure(A+\ker(\pi))$. Hence, there exist two nets $(x_\alpha)\subseteq A$ and $(m_\alpha)\subseteq\ker(\pi)$ such that $x_\alpha+m_\alpha\to x+\lambda m$. Since $\{f>\lambda\}$ is open by lower semicontinuity, it eventually follows from the translativity of $\rho_\pi$ that
\[
\lambda < f(x_\alpha+m_\alpha-\lambda m) \leq \rho_\pi(x_\alpha+m_\alpha-\lambda m) = \rho_\pi(x_\alpha)+\lambda \leq \lambda.
\]
Since this is impossible, we infer that $f\leq\rho_\pi^\ast$ must hold, concluding the proof.
\end{proof}

\smallskip

\begin{remark}
\label{rem: strategy}
(i) The preceding proposition shows that the dual representation in Lemma~\ref{theo: dual representation 2} and, hence, the multi-utility representation in Theorem~\ref{cor: second multi utility representation} can be equivalently stated in terms of the semicontinuous hulls of the functionals $\rho_\pi$ and $u_\pi$, respectively. This should be compared with the representation in Lemma 5.1 in Hamel and Heyde~\cite{HamelHeyde2010}.

\smallskip

(ii) The preceding proposition also suggests the following alternative path to establishing Lemma~\ref{theo: dual representation 2}: (1) Start with the representation in Lemma~\ref{theo: dual representation 1}. (2) Show that the functionals $\rho_\pi$ there can be replaced by their lower semicontinuous hulls $\lsc(\rho_\pi)$. (3) Show that we can discard from the representation all the functionals $\pi\in K_M^+\setminus\{0\}$ such that $\lsc(\rho_\pi)$ is not proper or, equivalently, $\ext(\pi)\cap\barr(A)=\emptyset$. (4) Use Proposition~\ref{prop: lsc hull} to replace the functionals $\lsc(\rho_\pi)$ with the more explicit functionals $\rho_\pi^\ast$. The advantage of the strategy pursued in the proof of Lemma~\ref{theo: dual representation 2} is that it avoids passing through semicontinuous hulls and the analysis of their properness.
\end{remark}

\smallskip

The representing functionals belonging to the multi-utility representation in Theorem~\ref{cor: second multi utility representation} are, by definition, upper semicontinuous. As a final step, we want to find conditions ensuring a multi-utility representation consisting of continuous functionals only. To achieve this, we exploit the link between the functionals $\rho_\pi$ and their regularizations $\rho_\pi^\ast$ established in Proposition~\ref{prop: lsc hull}.

\begin{lemma}
\label{lem: continuous representation}
Let $\pi\in K_M^+\setminus\{0\}$ be such that $\ext(\pi)\cap\barr(A)\neq\emptyset$. If $\Interior(A)\neq\emptyset$ and $\rho_\pi(x)<\infty$ for every $x\in L$, then $\rho_\pi$ is finite valued and continuous. In particular, $\rho_\pi=\rho^\ast_\pi$.
\end{lemma}
\begin{proof}
First of all, we claim that $\rho_\pi(x)>-\infty$ for every $x\in L$. To see this, take any functional $\psi\in\ext(\pi)\cap\barr(A)$ and note that for every $x\in L$
\[
\rho_\pi(x) \geq \rho^\ast_\pi(x) \geq \sigma_A(\psi)-\psi(x) > -\infty.
\]
As a result, $\rho_\pi$ is finite valued. Note that, by definition, $\rho_\pi$ is bounded above on $A$ by $0$. Since $A$ has nonempty interior and $\rho_\pi$ is convex, we infer from Theorem 8 in Rockafellar~\cite{Rockafellar1974} that $\rho_\pi$ is continuous. The last statement is a direct consequence of Proposition~\ref{prop: lsc hull}.
\end{proof}

\smallskip

The following multi-utility representation with continuous utility functionals is a direct consequence of Theorem~\ref{cor: second multi utility representation} and Lemma~\ref{lem: continuous representation}.

\begin{theorem}
\label{cor: third multi utility representation}
Assume that $\Interior(A)\neq\emptyset$ and that $\rho_\pi(x)<\infty$ for all $\pi\in K_M^+\setminus\{0\}$ with $\ext(\pi)\cap\barr(A)\neq\emptyset$ and $x\in L$. Then, the preference $\succeq_R$ can be represented by the multi-utility family
\[
\cU^{\ast\ast}=\{u_\pi \,: \ \pi\in K_M^+\setminus\{0\}, \ \ext(\pi)\cap\barr(A)\neq\emptyset\}.
\]
In addition, every element of $\cU^{\ast\ast}$ is finite valued and continuous.
\end{theorem}

\smallskip

We conclude by showing a number of sufficient conditions for the finiteness assumption in Lemma~\ref{lem: continuous representation} to hold. This should be compared with the results in Section 3 in Farkas et al.~\cite{FarkasKochMunari2015}. The recession cone of $A$ is denoted by
\[
\rec(A) := \{x\in L \,: \ x+y\in A, \ \forall y\in A\}.
\]
Note that $\rec(A)$ is the largest convex cone such that $A+\rec(A)\subseteq A$. In particular, if $A$ is a cone, then $\rec(A)=A$. Moreover, for any convex cone $C\subseteq L$ we denote by
\[
\qint(C) := \{x\in C \,: \ \psi(x)>0, \ \forall \psi\in C^+\setminus\{0\}\}
\]
the quasi interior of $C$. Note that we always have $\Interior(C)\subseteq\qint(C)$.

\begin{proposition}
\label{prop: conditions for continuous representation}
Let $\pi\in K_M^+\setminus\{0\}$ satisfy $\ext(\pi)\cap\barr(A)\neq\emptyset$. Then, $\rho_\pi(x)<\infty$ for every $x\in L$ if any of the following conditions holds:
\begin{enumerate}[(i)]
  \item $L=A+M$.
  \item $M\cap\qint(K)\neq\emptyset$.
  \item $M\cap\qint(\rec(A))\neq\emptyset$.
\end{enumerate}
\end{proposition}
\begin{proof}
The desired assertion clearly holds under (i). Since $K\subseteq\rec(A)$ by assumption (A1), we see that $\qint(K)\subseteq\qint(\rec(A))$. Hence, it suffices to establish that (iii) implies the desired assertion. So, assume that (iii) holds and take $m\in M\cap\qint(\rec(A))$. If $\rho_\pi(x)=\infty$ for some $x\in L$, then we must have $(x+M)\cap A=\emptyset$. It follows from a standard separation result, see e.g.\ Theorem 1.1.3 in Z\u{a}linescu~\cite{Zalinescu2002}, that we find a nonzero functional $\psi\in L'$ satisfying $\psi(x+\lambda m) \leq \sigma_A(\psi)$ for every $\lambda\in\R$. This is only possible if $\psi(m)=0$, which cannot hold because $\psi\in\barr(A)\subseteq(\rec(A))^+$. As a result, we must have $\rho_\pi(x)<\infty$ for every $x\in L$.
\end{proof}


\section{Applications}
\label{sect: applications}

In this final section we specify the general dual representation of $R$ to a number of concrete situations. The explicit formulation of the corresponding multi-utility representation can be easily derived as in Theorem~\ref{cor: second multi utility representation} and Theorem~\ref{cor: third multi utility representation}. Throughout the section we consider a probability space $(\Omega,\cF,\probp)$ and fix an index $d\in\N$. For every $p\in[0,\infty]$ and every Borel measurable set $S\subseteq\R^d$ we denote by $L^p(S)$ the set of all equivalence classes with respect to almost-sure equality of $d$-dimensional random vectors $X=(X_1,\dots,X_d):\Omega\to\R^d$ with $p$-integrable components such that $\probp[X\in S]=1$. As usual, we never explicitly distinguish between an equivalence class in $L^p(S)$ and any of its representative elements. We treat $\R^d$ as a linear subspace of $L^p(\R^d)$. For all vectors $a,b\in\R^d$ we set
\[
\langle a,b\rangle := \sum_{i=1}^da_ib_i.
\]
The expectation with respect to $\probp$ is simply denoted by $\E$. For every $p\in[1,\infty]$ the space $L^p(\R^d)$ can be naturally paired with $L^q(\R^d)$ for $q=\frac{p}{p-1}$ via the bilinear form
\[
(X,Z)\mapsto\E[\langle X,Z\rangle].
\]
Here, we adopt the usual conventions $\frac{1}{0}:=\infty$ and $\frac{\infty}{\infty}:=1$. Finally, for every random vector $X\in L^1(\R^d)$ we use the compact notation
\[
\E[X] := (\E[X_1],\dots,\E[X_d]).
\]


\subsection{Set-valued risk measures in a multi-currency setting}

We consider a financial market where $d$ different currencies are traded. Every element of $L^1(\R^d)$ is interpreted as a vector of capital positions expressed in our different currencies at some future point in time. For a pre-specified acceptance set $\cA\subseteq L^1(\R^d)$ we look for the currency portfolios that have to be set up at the initial time to ensure acceptability.


\subsubsection{The static case}

As a first step, we consider a one-period market with dates $0$ and $1$. In this setting, we focus on the currency portfolios that we have to build at time $0$ in order to ensure acceptability of currency positions at time $1$. This naturally leads to defining the set-valued map $R:L^1(\R^d)\rightrightarrows\R^d$ by
\[
R(X) := \{m\in\R^d \,: \ X+m\in\cA\}.
\]

\smallskip

\begin{assumption}
In this subsection we work under the following assumptions:
\begin{enumerate}
  \item[(1)] $\cA$ is norm closed, convex, and satisfies $\cA+L^1(\R^d_+)\subseteq\cA$.
  \item[(2)] $R(X)\notin\{\emptyset,\R^d\}$ for some $X\in L^1(\R^d)$.
\end{enumerate}
\end{assumption}

\smallskip

We derive the following representation by applying our general results to
\[
(L,L',K,A,M) = \Big(L^1(\R^d),L^\infty(\R^d),L^1(\R^d_+),\cA,\R^d\Big).
\]
This result should be compared with the dual representation established in Jouini et al.~\cite{JouiniMeddebTouzi2004}, Kulikov~\cite{Kulikov2008}, and Hamel and Heyde~\cite{HamelHeyde2010}.

\begin{proposition}
\label{prop: application 1}
For every $X\in L^1(\R^d)$ the set $R(X)$ can be represented as
\[
R(X) = \bigcap_{w\in\R^d_+\setminus\{0\}}\bigcap_{Z\in L^{\infty}(\R^d_+),\,\E[Z]=w}\{m\in\R^d \,: \ \langle m,w\rangle\geq\sigma_\cA(Z)-\E[\langle X,Z\rangle]\}.
\]
In addition, if $\cA$ is a cone, then we can simplify the above representation using that
\[
\sigma_\cA(Z)=
\begin{cases}
0 & Z\in\cA^+,\\
-\infty & \mbox{otherwise}.
\end{cases}
\]
\end{proposition}
\begin{proof}
Note that $K_M^+$ can be identified with $\R^d_+$ and that $\barr(\cA)$ is contained in $L^\infty(\R^d_+)$ by assumption (1). Since, for all $w\in\R^d$ and $Z\in L^\infty(\R^d)$, the random vector $Z$ (viewed as a functional on $L^1(\R^d)$) is an extension of $w$ (viewed as a functional on $\R^d$) precisely when $\E[Z]=w$, the desired representation follows immediately from Lemma~\ref{theo: dual representation 2}.
\end{proof}

\smallskip

\begin{remark}
Note that, in the above framework, the set $\qint(K)$ consists of all the random vectors in $L^1(\R^d)$ with components that are strictly positive almost surely and, hence, $M\cap\qint(K)\neq\emptyset$. This can be used to ensure multi-utility representations with continuous representing functionals; see Proposition~\ref{prop: conditions for continuous representation}.
\end{remark}

\smallskip

\begin{example}[{\bf Multidimensional Expected Shortfall}]
For every $X\in L^1(\R)$ and every $\alpha\in(0,1)$ we denote by $\ES_\alpha(X)$ the Expected Shortfall of $X$ at level $\alpha$, i.e.
\[
\ES_\alpha(X) := -\frac{1}{\alpha}\int_0^\alpha q_X(\beta)d\beta,
\]
where $q_X$ is any quantile function of $X$. The multi-dimensional acceptance set based on Expected Shortfall introduced in Hamel et al.~\cite{HamelRudloffYankova2013} is given by
\[
\cA = \{X\in L^1(\R^d) \,: \ \ES_{\alpha_i}(X_i)\leq0, \ \forall i\in\{1,\dots,d\}\}
\]
for a fixed $\alpha=(\alpha_1,\dots,\alpha_d)\in(0,1)^d$. Note that assumptions (1) and (2) hold. In particular, we have $R(0)=\R^d_+$. In addition, $\cA$ is a cone. Note that
\[
\cZ^\infty_w(\alpha) := \{Z\in L^\infty(\R^d) \,: \ \E[Z]=w\}\cap\cA^+ =
\left\{Z\in L^\infty(\R^d_+) \,: \ \E[Z]=w, \ Z\leq\frac{w}{\alpha}\right\}
\]
for every $w\in\R^d_+$ (where $\frac{w}{\alpha}$ is understood component by component). This follows from the standard dual representation of Expected Shortfall; see Theorem 4.52 in F\"{o}llmer and Schied~\cite{FoellmerSchied2016}. As a result, the dual representation in Proposition~\ref{prop: application 1} reads
\[
R(X) = \bigcap_{w\in\R^d_+\setminus\{0\}}\bigcap_{Z\in\cZ^\infty_w(\alpha)}\{m\in\R^d \,: \ \langle m,w\rangle\geq-\E[\langle X,Z\rangle]\}
\]
for every random vector $X\in L^1(\R^d)$.
\end{example}


\subsubsection{The dynamic case}

As a next step, we consider a multi-period financial market with dates $t=0,\dots,T$ and information structure represented by a filtration $(\cF_t)$ satisfying $\cF_0=\{\emptyset,\Omega\}$ and $\cF_T=\cF$. In this setting, currency portfolios can be rebalanced through time. A (random) portfolio at time $t\in\{0,\dots,T\}$ is represented by an $\cF_t$-measurable random vector in $L^0(\R^d)$. We denote by $\cC_t$ the set of $\cF_t$-measurable portfolios that can be converted into portfolios with nonnegative components by trading at time $t$. This means that, for all $\cF_t$-measurable portfolios $m_t$ and $n_t$, we can exchange $m_t$ for $n_t$ at time $t$ provided that $m_t-n_t \in \cC_t$. The sets $\cC_t$ are meant to capture potential transaction costs. A flow of portfolios is represented by an adapted process $(m_t)$. More precisely, for every date $t\in\{0,\dots,T-1\}$, the portfolio $m_t$ is set up at time $t$ and held until time $t+1$. The portfolio flows belonging to the set
\[
\cC := \{(m_t) \,: \ m_t-m_{t+1} \in \cC_{t+1}, \ \forall t\in\{0,\dots,T-1\}\}
\]
are said to be admissible. The admissibility condition is a direct extension of the standard self-financing property in frictionless markets.

\medskip

We look for all the initial portfolios that can be rebalanced in an admissible way until the terminal date in order to ensure acceptability. This leads to the set-valued map $R:L^1(\R^d)\rightrightarrows\R^d$ defined by
\begin{align*}
R(X)
&:=
\{m\in\R^d \,: \ \exists (m_t)\in\cC, \ n_T\in L^0(\R^d) \,:\, \\
&\phantom{a}\hspace{2cm}
m-m_0\in\cC_0, \ m_T-n_T\in\cC_T, \ X+n_T\in\cA\}.
\end{align*}
In words, the above set consists of all the initial portfolios that give rise, after a convenient exchange at date $0$, to an admissible rebalancing process making the outstanding currency position acceptable after a final portfolio adjustment at time $T$. This setting can be embedded in our framework because we can equivalently write
\[
R(X) = \bigg\{m\in\R^d \,: \ X+m\in\cA+\sum_{t=0}^{T}\cC_t\bigg\}.
\]

\smallskip

\begin{assumption}
In this subsection we work under the following assumptions:
\begin{enumerate}
  \item[(1)] $\cA$ is norm closed, convex, and satisfies $\cA+L^1(\R^d_+)\subseteq\cA$.
  \item[(2)] $\cC_t$ is convex and contains $L^0_t(\R^d_+)$ for every $t\in\{0,\dots,T\}$.
  \item[(3)] $(\cA+\sum_{t=0}^{T}\cC_t)\cap L^1(\R^d)$ is norm closed.
  \item[(4)] $R(X)\notin\{\emptyset,\R^d\}$ for some $X\in L^1(\R^d)$.
\end{enumerate}
\end{assumption}

\smallskip

We derive the following representation by applying our general results to
\[
(L,L',K,A,M) = \left(L^1(\R^d),L^\infty(\R^d),L^1(\R^d_+),\bigg(\cA+\sum_{t=0}^{T}\cC_t\bigg)\cap L^1(\R^d),\R^d\right).
\]
For convenience, we also set
\[
\cC^1_{1:T} := \bigg(\sum_{t=1}^{T}\cC_t\bigg)\cap L^1(\R^d).
\]
For later use note that
\[
\barr(\cC^1_{1:T}) \subseteq \barr\left(\sum_{t=1}^T\Big(\cC_t\cap L^1(\R^d)\Big)\right) = \bigcap_{t=1}^T\barr\Big(\cC_t\cap L^1(\R^d)\Big).
\]
The next result should be compared with the dual representation established in Hamel et al.~\cite{HamelHeydeRudloff2011} in the special setting of Example~\ref{ex: superreplication}.

\begin{proposition}
\label{prop: application 2}
For every $X\in L^1(\R^d)$ the set $R(X)$ can be represented as
\[
R(X) = \bigcap_{w\in\R^d_+\setminus\{0\}}\bigcap_{Z\in L^\infty(\R^d_+),\,\E[Z]=w}\{m\in\R^d \,: \ \langle m,w\rangle\geq\sigma_{\cA,\cC}(Z)-\E[\langle X,Z\rangle]\}.
\]
where we have set for every $Z\in L^\infty(\R^d)$
\[
\sigma_{\cA,\cC}(Z) := \sigma_\cA(Z)+\sigma_{\cC_0}(Z)+\sigma_{\cC^1_{1:T}}(Z)
\]
In addition, if $\cA$ is a cone, the above representation can be simplified by using that
\[
\sigma_\cA(Z)=
\begin{cases}
0 & Z\in\cA^+,\\
-\infty & \mbox{otherwise}.
\end{cases}
\]
Moreover, if $\cC_0$ is a cone, then
\[
\sigma_{\cC_0}(Z)=
\begin{cases}
0 & Z\in L^\infty(\R^d_+), \ \E[Z]\in\cC_0^+,\\
-\infty & \mbox{otherwise}.
\end{cases}
\]
Similarly, if $\cC_t$ is a cone for every $t\in\{1,\dots,T\}$, then
\[
\sigma_{\cC^1_{1:T}}(Z)=
\begin{cases}
0 & Z\in(\cC^1_{1:T})^+\subseteq\bigcap_{t=1}^T(\cC_t\cap L^1(\R^d))^+,\\
-\infty & \mbox{otherwise}.
\end{cases}
\]
\end{proposition}
\begin{proof}
The assertion follows from Proposition~\ref{prop: application 1} because $\cC_0\subset\R^d$ implies
\[
\bigg(\cA+\sum_{t=0}^{T}\cC_t\bigg)\cap L^1(\R^d) = \cA+\cC_0+\cC_{1:T}^1
\]
and observe that $\sigma_{\cA+\cC_0+\cC_{1:T}^1}=\sigma_\cA+\sigma_{\cC_0}+\sigma_{\cC^1_{1:T}}$.
\end{proof}

\smallskip

\begin{remark}
Note that, as in the static case, we have $M\cap\qint(K)\neq\emptyset$. This can be used to ensure multi-utility representations with continuous representing functionals; see Proposition~\ref{prop: conditions for continuous representation}.
\end{remark}

\smallskip

\begin{example}[{\bf Superreplication under proportional transaction costs}]
\label{ex: superreplication}
We adopt the discrete version of the model by Kabanov~\cite{Kabanov1999}. For every $t\in\{0,\dots,T\}$ we say that a set-valued map $S:\Omega\rightrightarrows\R^d$ is $\cF_t$-measurable provided that
\[
\{\omega\in\Omega \,: \ S(\omega)\cap\cU\neq\emptyset\}\in\cF_t
\]
for every open set $\cU\subset\R^d$. In this case, we denote by $L^0(S)$ the set of all random vectors $X\in L^0(\R^d)$ such that $\probp[X\in S]=1$. This set is always nonempty if $S$ has closed values; see Corollary 14.6 in Rockafellar and Wets~\cite{RockafellarWets2009}. Now, let $K_t:\Omega\rightrightarrows\R^d$ be an $\cF_t$-measurable set-valued map such that $K_t(\omega)$ is a polyhedral convex cone (hence $K_t(\omega)$ is closed) containing $\R^d_+$ for every $\omega\in\Omega$ and set
\[
\cC_t = L^0(K_t).
\]
Moreover, we consider the worst-case acceptance set
\[
\cA = L^1(\R^d_+).
\]
Assumptions (1) and (2) are easily seen to be satisfied. Moreover, $\cA$ as well as each of the sets $\cC_t$ is a cone. As proved in Theorem 2.1 in Schachermayer~\cite{Schachermayer2004}, assumption (3) always holds under the so-called ``robust no-arbitrage'' condition. Finally, as $0\in R(0)$, assumption (4) holds if and only if $\R^d$ is not entirely contained in $\sum_{t=0}^T\cC_t$. Note also that $\cA^+=L^\infty(\R^d_+)$. As a result, Proposition~\ref{prop: application 2} yields
\[
R(X) = \bigcap_{w\in K_0^+\setminus\{0\}}\bigcap_{Z\in L^\infty(\R^d_+),\ Z\in(\cC_{1:T}^1)^+,\,\E[Z]=w}\{m\in\R^d \,: \ \langle m,w\rangle\geq-\E[\langle X,Z\rangle]\}
\]
for every $X\in L^1(\R^d)$. The dual elements $Z$ in the above representation can be linked to consistent pricing systems, see e.g.\ Schachermayer~\cite{Schachermayer2004}. To see this, note that, for every $t\in\{0,\dots,T\}$, the set-valued map $K^+_t:\Omega\rightrightarrows\R^d$ defined by
\[
K^+_t(\omega)=\big(K_t(\omega)\big)^+
\]
is $\cF_t$-measurable, see e.g.\ Exercise 14.12 in Rockafellar and Wets~\cite{RockafellarWets2009}, and such that
\[
\Big(\cC_t\cap L^1(\R^d)\Big)^+=L^0(K^+_t)\cap L^\infty(\R^d)
\]
by measurable selection, see the argument in the proof of Theorem 1.7 in Schachermayer~\cite{Schachermayer2004}. As a result, every dual element $Z$ in the above dual representation satisfies
\[
\E[Z]\in K_0^+, \ \ \ \ \ \ Z \in (\cC_{1:T}^1)^+ \subseteq \bigcap_{t=1}^T\Big(\cC_t\cap L^1(\R^d)\Big)^+ \subseteq \bigcap_{t=1}^TL^0(K^+_t).
\]
This shows that the $d$-dimensional adapted process $(\E[Z\vert\cF_t])$, where the conditional expectations are taken componentwise, satisfies $\E[Z\vert\cF_T]=Z$ and $\E[Z\vert\cF_t] \in L^0(K^+_t)$ for every $t\in\{0,\dots,T\}$ and thus qualifies as a consistent pricing system. In other words, the above dual elements $Z$ can be viewed as the terminal values of consistent pricing systems.
\end{example}

\smallskip

\begin{remark}
(i) It is worth noting that our approach provides a different path, compared to the strategy pursued in Schachermayer~\cite{Schachermayer2004}, to establish the existence of consistent pricing systems under the robust no-arbitrage assumption (admitting the closedness of the reference target set). Moreover, by rewriting the above dual representation in terms of consistent pricing systems, we recover the (localization to $L^1(\R^d)$ of the) superreplication theorem by Schachermayer~\cite{Schachermayer2004}.

\smallskip

(ii) The above dual representation was also obtained in Hamel et al.~\cite{HamelHeydeRudloff2011}. Differently from that paper, we have not derived it from the superreplication theorem in Schachermayer~\cite{Schachermayer2004} but from a direct application of our general results.
\end{remark}


\subsection{Systemic set-valued risk measures based on acceptance sets}

We consider a single-period economy with dates $0$ and $1$ and a financial system consisting of $d$ entities. Every element of $L^\infty(\R^d)$ is interpreted as a vector of capital positions of the various financial entities at time $1$. The individual positions can be aggregated through a function $\Lambda:\R^d\to\R$. For a pre-specified acceptance set $\cA\subseteq L^\infty(\R)$ we look for the cash injections at time $0$ that ensure the acceptability of the aggregated system. This leads to the set-valued map $R:L^\infty(\R^d)\rightrightarrows\R^d$ defined by
\[
R(X) := \{m\in\R^d \,: \ \Lambda(X+m)\in\cA\}.
\]
This setting can be easily embedded in our framework because we can write
\[
R(X) = \{m\in\R^d \,: \ X+m\in\Lambda^{-1}(\cA)\}.
\]

\smallskip

\begin{assumption}
In this subsection we work under the following assumptions:
\begin{enumerate}
  \item[(1)] $\cA$ is convex and satisfies $\cA+L^\infty(\R_+)\subseteq\cA$.
  \item[(2)] $\Lambda$ is nondecreasing and concave.
  \item[(3)] $\Lambda^{-1}(\cA)$ is $\sigma(L^\infty(\R^d),L^1(\R^d))$-closed.
  \item[(4)] $R(X)\notin\{\emptyset,\R^d\}$ for some $X\in L^\infty(\R^d)$.
\end{enumerate}
\end{assumption}

\smallskip

We derive the following representation by applying our general results to
\[
(L,L',K,A,M) = \Big(L^\infty(\R^d),L^1(\R^d),L^\infty(\R^d_+),\Lambda^{-1}(\cA),\R^d\Big).
\]
The result should be compared with the dual representations established in Ararat and Rudloff~\cite{AraratRudloff}.

\begin{proposition}
\label{prop: systemic}
For every $X\in L^\infty(\R^d)$ the set $R(X)$ can be represented as
\[
R(X) = \bigcap_{w\in\R^d_+\setminus\{0\}}\bigcap_{Z\in L^1(\R^d_+),\,\E[Z]=w}\{m\in\R^d \,: \ \langle m,w\rangle\geq\sigma_{\Lambda^{-1}(\cA)}(Z)-\E[\langle X,Z\rangle]\}.
\]
If $\Lambda^{-1}(\cA)$ is a cone, then we can simplify the above representation using that
\[
\sigma_{\Lambda^{-1}(\cA)}(Z)=
\begin{cases}
0 & Z\in(\Lambda^{-1}(\cA))^+,\\
-\infty & \mbox{otherwise}.
\end{cases}
\]
\end{proposition}
\begin{proof}
Note that $\Lambda^{-1}(\cA)$ is convex and satisfies $\Lambda^{-1}(\cA)+L^\infty(\R^d_+)\subseteq\Lambda^{-1}(\cA)$ by assumptions (1) and (2). Note also that $K_M^+$ can be identified with $\R^d_+$. In addition, we have $\barr(\Lambda^{-1}(\cA))\subseteq L^1(\R^d_+)$. Since, for all $w\in\R^d$ and $Z\in L^1(\R^d)$, the random vector $Z$ (viewed as a functional on $L^\infty(\R^d)$) is an extension of $w$ (viewed as a functional on $\R^d$) precisely when $\E[Z]=w$, the desired representation follows from Lemma~\ref{theo: dual representation 2}.
\end{proof}

\smallskip

\begin{example}[{\bf Weighted aggregated losses}]
Let $\alpha=(\alpha_1,\dots,\alpha_d)\in(1,\infty)^d$ and consider the aggregation function $\Lambda:\R^d\to\R$ defined by
\[
\Lambda(x) = \sum_{i=1}^d\max(x_i,0)+\sum_{i=1}^d\alpha_i\min(x_i,0).
\]
Moreover, define the acceptance set $\cA\subset L^\infty(\R)$ by
\[
\cA = \{X\in L^\infty(\R) \,: \ \E[X]\geq0\}.
\]
Clearly, assumptions (1) and (2) hold. Since $\cA$ is closed with respect to dominated almost-sure convergence and $\Lambda$ is continuous, we easily see that $\Lambda^{-1}(\cA)$ is also closed with respect to dominated almost-sure convergence. Hence, we can repeat the argument in Theorem 3.2 by Delbaen~\cite{Delbaen2002} to conclude that (3) holds; see also Proposition 3.2 in Arduca et al.~\cite{ArducaKochMunari2019}. Finally, note that $R(0)=\{m\in\R^d \,: \ \Lambda(m)\geq0\}$, showing that (4) holds. Now, observe that $\barr(\cA)=\R_+$. The concave conjugate of $\Lambda$ is the map $\Lambda^\bullet:\R^d\to[-\infty,\infty)$ defined by
\[
\Lambda^\bullet(z) := \inf_{x\in\R^d}\{\langle x,z\rangle-\Lambda(x)\}.
\]
It is easy to verify that
\[
\Lambda^\bullet(z) =
\begin{cases}
0 & \mbox{if} \ z\in[1,\alpha_1]\times\cdots\times[1,\alpha_d],\\
-\infty & \mbox{otherwise}.
\end{cases}
\]
Then, it follows from Proposition 3.13 and Lemma 3.20 in Arduca et al.~\cite{ArducaKochMunari2019} that
\[
\sigma_{\Lambda^{-1}(\cA)}(Z)=
\begin{cases}
0 & \mbox{if there exists $\lambda\in[0,\infty)$ such that $\lambda e\leq Z\leq\lambda\alpha$},\\
-\infty & \mbox{otherwise},
\end{cases}
\]
where $e=(1,\dots,1)\in\R^d$. Now, for notational convenience define
\[
\cZ^\infty(\alpha) := \{Z\in L^\infty(\R^d_+) \,: \ \exists \lambda\in[0,\infty) \,:\, \lambda e\leq Z\leq\lambda\alpha\}.
\]
As a result, Proposition~\ref{prop: systemic} yields
\[
R(X) = \bigcap_{w\in\R^d_+\setminus\{0\}}\bigcap_{Z\in\cZ^\infty(\alpha),\,\E[Z]=w}\{m\in\R^d \,: \ \langle m,w\rangle\geq-\E[\langle X,Z\rangle]\}
\]
for every random vector $X\in L^\infty(\R^d)$.
\end{example}


\begin{thebibliography}{99}

\bibitem{AraratRudloff} Ararat, \c{C}., Rudloff, B.: Dual representations for systemic risk measures, Math.\ Financ.\ Econ.\ 14, 139-174 (2020)

\bibitem{ArducaKochMunari2019} Arduca, M., Koch-Medina, P., Munari, C.: Dual representations for systemic risk measures based on acceptance sets, to appear in Math.\ Financ.\ Econ.\ (2020)

\bibitem{ArmentiCrepeyDrapeauPapapantoleon2018} Armenti, Y., Cr\'{e}pey, S., Drapeau, S., Papapantoleon, A.: Multivariate shortfall risk allocation and systemic risk, SIAM J.\ Financ.\ Math.\ 9, 90-126 (2018)

\bibitem{AubinEkeland1984} Aubin, J.-P., Ekeland, I.: Applied Nonlinear Analysis, Dover Publications (2006)

\bibitem{Aumann1962} Aumann, R.J.: Utility theory without the completeness axiom, Econometrica 30, 445-462 (1962)

\bibitem{BaesKochMunari2020} Baes, M., Koch-Medina, P., Munari, C.: Existence, uniqueness, and stability of optimal payoffs of eligible assets, Math.\ Financ.\ 30, 128-166 (2020)

\bibitem{BevilacquaBosiKaucicZuanon2019} Bevilacqua, P., Bosi, G., Kaucic, M., Zuanon, M.E.: Pareto optimality on compact spaces in a preference-based setting under incompleteness, Int.\ J.\ Uncertain.\ Fuzz.\ 27, 239-249 (2019)

\bibitem{Bewley2002} Bewley, T.F.: Knightian decision theory. Part I, Decis.\ Econ.\ Finance 25, 79-110 (2002)

\bibitem{Bosiherden2012} Bosi, G., Herden, G.: Continuous multi-utility representations of preorders, J.\ Math.\ Econ.\ 48, 212-218 (2012)

\bibitem{BosiEstevanZuanon2018} Bosi, G., Estevan, A., Zuanon, M.E.: Partial representations of orderings, Int.\ J.\ Uncertain.\ Fuzz.\ 26, 453-473 (2018)

\bibitem{CrespiHamelRoccaSchrage2018} Crespi, G., Hamel, A.H., Rocca, M., Schrage, C.: Set relations and approximate solutions in set optimization, arXiv:1812.03300 (2018)

\bibitem{Delbaen2002} Delbaen, F.: Coherent risk measures on general probability spaces, In: Sandmann, K., Sch\"{o}nbucher, P.J.\ (eds.), Advances in Finance and Stochastics: Essays in Honour of Dieter Sondermann, pp.\ 1-37, Springer (2002)

\bibitem{DrapeauKupper2013} Drapeau, S., Kupper, M.: Risk preferences and their robust representations, Math.\ Oper.\ Res.\ 38, 28-62 (2013)

\bibitem{DubraMaccheroniOk2004} Dubra, J., Maccheroni, F., Ok, E.A.: Expected utility theory without the completeness axiom, J.\ Econ.\ Theory 115, 118-133 (2004)

\bibitem{EliazOk2006} Eliaz, K., Ok, E.A.: Indifference or indecisiveness? Choice-theoretic foundations of incomplete preferences, Game.\ Econ.\ Behav.\ 56, 61-86 (2006)

\bibitem{Evren2008} Evren, \"{O}.: On the existence of expected multi-utility representations, Econ.\ Theor.\ 35, 575-592 (2008)

\bibitem{Evren2014} Evren, \"{O}: Scalarization methods and expected multi-utility representations, J.\ Econ.\ Theory 151, 30-63 (2014)

\bibitem{EvrenOk2011} Evren, \"{O}., Ok, E.A.: On the multi-utility representation of preference relations, J.\ Math.\ Econ.\ 47, 554-563 (2011)

\bibitem{FarkasKochMunari2014} Farkas, W., Koch-Medina P., Munari, C.: Beyond cash-additive risk measures: when changing the num\'{e}raire fails, Financ.\ Stoch.\ 18, 145-173 (2014)

\bibitem{FarkasKochMunari2015} Farkas, W., Koch-Medina, P., Munari, C.: Measuring risk with multiple eligible assets, Math.\ Financ.\ Econ.\ 9, 3–27 (2015)

\bibitem{FeinsteinRudloff2013} Feinstein, Z., Rudloff, B.: Time consistency of dynamic risk measures in markets with transaction costs, Quant.\ Financ.\ 13, 1473-1489 (2013)

\bibitem{FeinsteinRudloffWeber2017} Feinstein, Z., Rudloff, B., Weber, S.: Measures of systemic risk, SIAM J.\ Financ.\ Math.\ 8, 672-708 (2017)

\bibitem{FoellmerSchied2002} F\"{o}llmer, H., Schied, A.: Convex measures of risk and trading constraints, Financ.\ Stoch.\ 6, 429-447 (2002)

\bibitem{FoellmerSchied2016} F\"{o}llmer, H., Schied, A.: Stochastic Finance: An Introduction in Discrete Time, de Gruyter (2016)

\bibitem{FrittelliScandolo2006} Frittelli, M., Scandolo, G.: Risk measures and capital requirements for processes, Math.\ Financ.\ 16, 589-612 (2006)

\bibitem{GalaabaatarKarni2013} Galaabaatar, T., Karni, E.: Subjective expected utility with incomplete preferences, Econometrica 81, 255-284 (2013)

\bibitem{HaierMolchanovSchmutz2016} Haier, A., Molchanov, I., Schmutz, M.: Intragroup transfers, intragroup diversification and their risk assessment, Ann.\ Finance 12, 363-392 (2016)

\bibitem{HamelHeyde2010} Hamel, A.H., Heyde, F.: Duality for set-valued measures of risk, SIAM J.\ Financ.\ Math.\ 1, 66-95 (2010)

\bibitem{HamelHeydeRudloff2011} Hamel, A.H., Heyde F., Rudloff, B.: Set-valued risk measures for conical market models, Math.\ Financ.\ Econ.\ 5, 1-28 (2011)

\bibitem{HamelHeydeLoehneRudloffSchrage2015} Hamel, A.H., Heyde F., L\"{o}hne, A., Rudloff, B., Schrage, C.: Set optimization. A rather short introduction, In: Hamel A., Heyde F., L\"{o}hne A., Rudloff B., Schrage C.\ (eds), Set Optimization and Applications: The State of the Art, Springer (2015)

\bibitem{HamelRudloffYankova2013} Hamel, A.H., Rudloff, B., Yankova, M.: Set-valued average value at risk and its computation, Math.\ Financ.\ Econ.\ 7, 229-246 (2013)

\bibitem{JouiniMeddebTouzi2004} Jouini, E., Meddeb, M., Touzi, N.: Vector-valued coherent risk measures, Financ.\ Stoch.\ 8, 531-552 (2004)

\bibitem{Kabanov1999} Kabanov, Y.:  Hedging and liquidation under transaction costs in currency markets, Financ.\ Stoch.\ 3, 237-248 (1999)

\bibitem{Kaminski2007} Kaminski, B.: On quasi-orderings and multi-objective functions, Eur.\ J.\ Oper.\ Res.\ 177, 1591-1598 (2007)

\bibitem{Kulikov2008} Kulikov, A.V.: Multidimensional coherent and convex risk measures, Theory Probab.\ Appl.\ 52, 614-635 (2008)

\bibitem{Mandler2005} Mandler, M.: Incomplete preferences and rational intransitivity of choice, Game.\ Econ.\ Behav.\ 50, 255-277 (2005)

\bibitem{MolchanovCascos2016} Molchanov, I., Cascos, I.: Multivariate risk measures: a constructive approach based on selections, Math.\ Financ.\ 26, 867-900 (2016)

\bibitem{NishimuraOk2016} Nishimura, H., Ok, E.A.: Utility representation of an incomplete and nontransitive preference relation, J.\ Econ.\ Theory 166, 164-185 (2016)

\bibitem{Ok2002} Ok, E.A.: Utility representation of an incomplete preference relation, J.\ Econ.\ Theory 104, 429-449 (2002)

\bibitem{Rockafellar1974} Rockafellar, R.T.: Conjugate Duality and Optimization, Society for Industrial and Applied Mathematics (1974)

\bibitem{RockafellarWets2009} Rockafellar, R.T., Wets, R. J.-B.: Variational Analysis, Springer (2009)

\bibitem{Schachermayer2004} Schachermayer, W.: The fundamental theorem of asset pricing under proportional transaction costs in finite discrete time, Math.\ Financ.\ 14, 19-48 (2004)

\bibitem{Zalinescu2002} Z\v{a}linescu, C.: Convex Analysis in General Vector Spaces, World Scientific (2002)

\end{thebibliography}
\end{document}